\documentclass[12pt,a4paper]{amsart}

\usepackage{amsmath,amsfonts,amssymb}

\newcommand{\C}{\mathbb C}

\def\V{{\mathcal{V}}}

\def\A{{\mathcal{A}}}

\def\d{\partial}
\def\nd{\mathsf{n.d.}}
\def\H{\mathcal{H}}

\def\CP1{\mathbb{C}\mathrm{P}^1}

\newtheorem{theorem}{Theorem}[section]
\newtheorem{proposition}[theorem]{Proposition}

\theoremstyle{definition}

\newtheorem{example}[theorem]{Example}
\newtheorem{remark}[theorem]{Remark}

\title[On deformations of Hurwitz Frobenius manifolds]
{A remark on deformations of Hurwitz Frobenius manifolds}

\thanks{A.~B is partially supported by the grants RFBR-07-01-00593, NSh-709.2008.1. Both A.~B. and S.~S. are partly supported by the Vidi grant of NWO}

\author{A.~Buryak}

\address{A.~Buryak:\newline
Department of Mathematics,
University of Amsterdam, \newline
P.~O.~Box 94248, 1090 GE Amsterdam, 
The Netherlands\newline 
\indent and\newline
Department of Mathematics, Moscow State University,\newline
Leninskie gory, 19992 GSP-2 Moscow, Russia}
\email{a.y.buryak@uva.nl, buryaksh@mail.ru}

\author{S.~Shadrin}

\address{S.~Shadrin:\newline
Department of Mathematics,
University of Amsterdam, \newline
P.~O.~Box 94248, 1090 GE Amsterdam, 
The Netherlands\newline 
\indent and\newline
Department of Mathematics, Institute of System Research,\newline
Nakhimovsky prospekt 36-1, Moscow 117218, Russia}
\email{s.shadrin@uva.nl, shadrin@mccme.ru}

\begin{document}

\begin{abstract}
In this note we use the formalism of multi-KP hierarchies in order to give some general formulas for infinitesimal deformations of solutions of the Darboux-Egoroff system. As an application, we explain how Shramchenko's deformations of Frobenius manifold structures on Hurwitz spaces fit into the general formalism of Givental-van de Leur twisted loop group action on the space of semi-simple Frobenius manifolds.
\end{abstract}

\maketitle

\tableofcontents

\section{Introduction}

In~\cite{Dub}, Dubrovin has associated a structure of Frobenius manifolds to an arbitrary Hurwitz space of meromorphic functions on Riemann surfaces of genus $g$ with simple finite critical values and a prescribed ramification indices over infinity. 
Shramchenko observed~\cite{Shr2} that the structure of Frobenius manifold associated to a Hurwitz space can be included into a family of Frobenius manifold structures parametrized by a symmetric $g\times g$ matrix. There is a beautiful description of this deformation in terms of the values of holomorphic differentials at the critical points and their $B$-periods matrix.

Meanwhile, Givental in~\cite{Giv1,Giv2} and, independently, van de Leur in~\cite{Leu}  have constructed an action of the twisted loop group of $GL_n$ on the space of semi-simple Frobenius manifolds. Moreover, Givental has shown that this group acts transitively on the space semi-simple Frobenius manifolds. This two constructions of the group action were identified in~\cite{FeiLeuSha} via an identification of the formulas of Y.-P.~Lee for the infinitesimal Givental action~\cite{Lee} with the tangent van de Leur action computed in~\cite{FeiLeuSha} in terms of twisted wave functions of multi-component KP hierarchies. 

In this paper, we extend in some way the formulas for the tangent van de Leur action computed in~\cite{FeiLeuSha}. Namely, we express infinitesimal Lie algebra action on the space of solutions of the Darboux-Egoroff system in terms of the twisted wave functions of multi-component KP. In principle, these formulas are of independent interest. In particular, they allow us to fit Shramchenko's deformations into a general Givental-van de Leur scheme. In particular, it is interesting to trace a corrsepondence between geometric ingridients of Shramchenko's deformation and particular wave functions of the multi-component KP hierarchy that is associated to an arbitrary solution of the Darboux-Egoroff system in van de Leur's approach .

\subsection{Organization of the paper} In section~\ref{sec1}, we recall the constructions of Hurwitz Frobenius manifolds and their deromations.
In section~\ref{sec2}, we recall the van de Leur approach to Frobenius manifols and use it in order to derive explicit formulas for the Givental-van de Leur infinitesimal deformations of solutions of the Darboux-Egoroff equations. In section~\ref{sec3} we discuss the simplest possible example of such infinitesimal deformations that can be integrated explicitely and show that it is exactly the way one could obtain Shramchenko's deformations of Hurwitz Frobenius manifolds.

\subsection{Acknowledgements}
The authors are grateful to H.~Posthuma for a useful discussion.

\section{Frobenius structures associated to Hurwitz spaces} \label{sec1}

\subsection{Darboux-Egoroff equations}

In this paper we consider only semi-simple Frobenius manifolds. There is a way to encode the structure of a semi-simple Frobenius manifold in canonical coordinates as a solution of a system of PDEs that is called Darboux-Egoroff equations~\cite{Dub}. 

Let $n\geq 1$. We consider functions $\gamma_{ij}=\gamma_{ji}$, $i,j=1,\dots,n$, $i\not=j$, in variables $u_1,\dots,u_n$. The Darboux-Egoroff equations read:
\begin{align} \label{eq:DE}
\frac{\d\gamma_{ij}}{\d u_k} & = \gamma_{ik}\gamma_{kj}, & & i\not= j\not= k\not= i \\
\sum_{k=1}^n \frac{\d\gamma_{ij}}{\d u_k} & =0 & & i\not= j \notag
\end{align}

Is it convenient to collect $\gamma_{ij}$ into a symmetric matrix with the diagonal terms that can be either equal to $0$ or just arbitrary. We introduce a special notation for that. Let $M$ be a symmetric matrix. By $\nd M$ we denote the same matrix with non-specified diagonal terms.

\subsection{Hurwitz spaces} 

We fix some integer numbers $a_1,\dots,a_m>0$ and $g\geq 0$. Let $\H$ be the space of the equivalence classes of the tuples of data $(C_g,\{a_i,b_i\}_{i=1}^g,f\colon C_g\to\CP1)$, where $C_g$ is a Riemann surface of genus $g$, $\{a_i,b_i\}_{i=1}^g$ is a choice of the canonical basis of cycles on $C_g$, and $f\colon C_g\to\CP1$ is a meromorphic function of degree $d:=\sum_{i=1}^m a_m$ with exactly $m$ poles of multiplicity $a_1,\dots,a_m$ and $n:=2g+d+m-2$ simple critical points $x_1,\dots,x_{n}\in C_g$. In addition, we choose local parameters $z_,\dots,z_n$ at the points $x_1,\dots,x_{n}\in C_g$ such that $f=z_i^2$ in a neighbourhood of $x_i$. Two tuples of this data are equivalent if there is a biholomorphic map between two source curves that preserves the rest of the data.

The critical values of meromorphic functions $u_i:=f(x_i)$, $i=1,\dots,n$, are local coordinates on the space $\H$. 

We recall the Kokotov-Korotkin construction~\cite{KokKor} of a solution of the Darboux-Egoroff equations. Let $W(P,Q)$ be the canonical meromorphic bidifferential on a Riemann surface $C_g$. That is, $W(P,Q)$ is specified by the following properties: it is symmetric, it has a quadratic pole on the diagonal $P = Q$ with biresidue $1$, and its $a$-periods with respect to both variables vanish. Then the functions
\begin{equation}\label{eq:KokKor}
\gamma_{ij}:=\frac{1}{2}W(x_i,x_j):=\left.\frac{1}{2}\frac{W(P,Q)}{dz_i(P)dz_j(Q)}\right|_{P=x_i,Q=x_j}
\end{equation}
in variables $u_1,\dots,u_n$ satisfy the Darboux-Egoroff equations.

\subsection{Shramchenko's deformations}\label{sec:shr}

Let $\omega_i$, $i=1,\dots,n$, be the basis of holomorphic differentials on $C_g$ normalized by $\int_{a_i}\omega_j=\delta_{ij}$. Denote by $\omega$ the matrix of the values of $\omega_i$ at critical points: 
\begin{equation}\label{eq:omega-subs}
\omega_{ij}:=\omega_i(x_j):=\left.\frac{\omega_i(P)}{dz_j(P)}\right|_{P=x_j}
\end{equation}
Denote by $B$ the matrix of $b$-periods of these differentials divided by $\pi\sqrt{-1}$, $B_{ij}:=\frac{1}{\pi \sqrt{-1} }\int_{b_i}\omega_j$. Let $M$ be an arbitrary $g\times g$ symmetric matrix such that $B+M$ is non-degenerate. Shramchenko's deformations of Hurwitz Frobenius manifolds~\cite{Shr2} are given by the formula
\begin{equation} \label{eq:shr}
\nd\gamma(M):= \nd \left(\gamma - \omega^t (B+M)^{-1} \omega\right).
\end{equation}
Here $\nd\gamma$ is given by equation~\eqref{eq:KokKor}. Shramchenko proved that $\gamma_{ij}(M)$ are solutions of the Darboux-Egoroff equations in the variables $u_1,\dots,u_n$ in the domain $\det(B+M)\not=0$. Observe that $\nd\gamma(M)$ tends to $\nd\gamma$ when $(B+M)^{-1}$ tends to zero.

The proof that $\nd\gamma(M)$ is  a solution of the Darboux-Egoroff equations is based on Rauch variational formula and its corollaries:
\begin{align} \label{eq:Rauch1}
\frac{\d W(P,Q)}{\d u_j} & = \frac{1}{2}W(P,x_j)W(Q,x_j), \\ \label{eq:Rauch2}
\frac{\d \omega_i(P)}{\d u_j} & = \frac{1}{2}\omega_i(x_j)W(P,x_j), \\ \label{eq:Rauch3}
\frac{\d B_{kl}}{\d u_j} & = \omega_k(x_j)\omega_l(x_j),
\end{align}
where evaluation of differentials at particular points is defined in~\eqref{eq:KokKor} and~\eqref{eq:omega-subs}.


\section{Van de Leur's formalism for Frobenius manifolds}\label{sec2}

In this section we explain van de Leur's construction of a Frobenius structure associated to a point in the isotropic semi-infinite Grassmannian.

\subsection{Basic definitions}
Let $V=\langle e_1,\dots,e_n\rangle$ be an $n$-dimensional vector space over $\C$. Let $z$ be a formal variable. We denote by $\V$  
the vector space $\Lambda^{\infty/2}\left(V\otimes \C[z^{-1},z]\right)$ spanned by the semi-infinite wegde products 
$$
\omega=(e_{i_1}z^{d_1})\wedge(e_{i_2}z^{d_2})\wedge(e_{i_3}z^{d_3})\wedge\dots
$$
such that the tail of $\omega$ coinsides with the tail of vacuum vector
$$
|0\rangle : = (e_1z^0)\wedge\cdots\wedge (e_nz^0)\wedge(e_1z^1)\wedge\cdots\wedge (e_nz^1)\wedge\dots.
$$
By tail of $\omega$ we call another basis vector in $\V$ that is obtained from $\omega$ by removing the first few factors in the wedge product.

Consider a matrix series $A(z)\in End(V)\otimes \C[[z^{-1},z]]$ such that $A^t(-z)A(z)=\mathrm{Id}\cdot z^0$ (it is better to imagine it as a finite product of invertible matrix series in $End(V)\otimes \C[[z]]$ and $End(V)\otimes \C[[z^{-1}]]$ satisfying the same symplectic condition). 

Let $\alpha_i$ be a local Lie algebra element whose action on $V\otimes \C[z^{-1},z]$ is defined by 
$$
\alpha_i(e_jz^d):=\left\{
\begin{array}{ll}
 e_j z^{d+1} & \text{if } i=j, \\
 0 & \text{otherwise,}
\end{array}
\right.
$$
and is expanded to $\V$ by the Leibnitz rule.

All basic objects that we are going to consider are some matrix elements of the operator 
$$
\mathcal{A}:=\exp(\sum_{i=1}^n \alpha_i u_i) A(z),
$$ 
where $u_1,\dots,u_n$ are formal variables.

We denote by $\gamma_{ij}=\gamma_{ij}(A)$, $i,j=1,\dots,n$, $i\not= j$, the following matrix elements of $\A$:
$$
\gamma_{ij}:= \pm \frac{\left< |0\rangle \left| \A \left| (e_iz^{-1})\wedge \d_{(e_jz^0)} |0\rangle \right> \right. \right.}
{\left< |0\rangle \left| \A \left| |0\rangle \right> \right. \right.} .
$$
(the vector $(e_iz^{-1})\wedge \d_{(e_jz^0)} |0\rangle$ is obtained, up to a sign, from the vacuum vector $|0\rangle$ by the replacement of $(e_jz^0)$ by $(e_iz^{-1})$).

We denote by $(\Psi_d)_{ij}=(\Psi_d)_{ij}(A)$, $i,j=1,\dots,n$, the following matrix elements of $\A$:
$$
(\Psi_d)_{ij}:= \frac{\left< (e_jz^{-1})\wedge |0\rangle \left| \A \left| (e_iz^{-1-d})\wedge |0\rangle \right> \right. \right.}
{\left< |0\rangle \left| \A \left| |0\rangle \right> \right. \right.} .
$$
These matrices are can be arranged into a generating series $\Psi(z):=\sum_{d=1}^\infty z^d \Psi_d$ that would be a wave function of multi-KP hierarchy multiplied by $A(z)$ from the right. The property $A^t(-z)A(z)=\mathrm{Id}\cdot z^0$ imply that $\Psi^t(-z)\Psi(z)=\mathrm{Id}\cdot z^0$.

Van de Leur has shown in~\cite{Leu} that a formal locally semi-simple Frobenius structure can expressed in terms of the matrices $\gamma$ and $\Psi_d$, $d\geq 0$. In particular, $\gamma_{ij}$ is a solution of the Darboux-Egoroff system; $u_1,\dots,u_n$ are canonical coordinates; $\Psi_0^t\Psi_1 \mathbf{1}$ is a column of flat coordinates (here $\mathbf{1}$ is the column of units); and $(1/2)\cdot \mathbf{1}^t \Psi_0^t (-\Psi_3\Psi_0^t + \Psi_2\Psi_1^t) \Psi_0 \mathbf{1}$ is the prepotential of a Frobenius manifold.

\subsection{Infenitesimal deformations}

From the previous section we see that there is an action of the groups of matrices $A(z)\in Hom(V,V)\otimes \C[[z]]$, $A^t(-z)A(z)=\mathrm{Id}\cdot z^0$, and $A(z^{-1})\in Hom(V,V)\otimes \C[[z^{-1}]]$, $A^t(-z^{-1})A(z^{-1})=\mathrm{Id}\cdot z^0$. This group action is crucially important, see, e.g.,~\cite{FeiLeuSha} for a list of references for particular applications.

We discuss the corresponding Lie algebra action. Let $k\geq 0$ and $\ell>0$. Let matrices $r$ and $s$ be symmetric for odd $\ell$ and skewsymmetric for even $\ell$. It is proven in~\cite{FeiLeuSha} that
\begin{align}\label{eq:r-psi}
& \left.\frac{\d}{\d \epsilon}\right|_{\epsilon=0}\Psi_k(A\exp{\epsilon(r z^{-\ell})})  = \\ 
& \Psi_{\ell+k}(A)r  
\notag  - \sum_{p=1}^\ell\sum_{q=0}^{\ell-p}(-1)^{\ell-p-q}\Psi_q(A) r \Psi_{\ell-p-q}^t(A)\Psi_{p+k}(A); \\
&\label{eq:l-psi} \left.\frac{\d}{\d \epsilon}\right|_{\epsilon=0}\Psi_k(A\exp{\epsilon(s z^{\ell})})  = 
\begin{cases}
\Psi_{k-\ell}(A) s, & \ell\leq k; \\
0, & \ell >k
\end{cases} 
\end{align}
This allows to compute the action of this Lie algebra on the prepotential of Frobenius manifolds in flat coordinates, since both the prepotential and the flat coordinates are expressed in terms of $\Psi_d$, $d\geq 0$.

In order to deal with Hurwitz Frobenius manifolds, we need some formulas for the Lie algebra action on $\nd\gamma(A)$.
\begin{theorem} Let $\ell\geq 0$. Let matrices $r$ and $s$ be symmetric for odd $\ell$ and skewsymmetric for even $\ell$. We have:
\begin{align}
\label{eq:r-gamma}
\left.\frac{\d}{\d \epsilon}\right|_{\epsilon=0}\nd\gamma(A\exp{\epsilon(r z^{-\ell})}) & = \nd\sum_{i+j=\ell-1} (-1)^{j-1} \Psi_i(A)r\Psi_j^t(A) \\ 
\left.\frac{\d}{\d \epsilon}\right|_{\epsilon=0}\nd\gamma(A\exp{\epsilon(s z^{\ell})}) & = 0.  
\label{eq:s-gamma}
\end{align}
\end{theorem}

\begin{proof} This theorem is an easy consequence of formulas~\eqref{eq:r-psi} and~\eqref{eq:l-psi}. Indeed, there is a relation between $\Psi_d$, $d\geq 0$ and $\nd\gamma$ that is proven in~\cite{Leu}. For any $d\geq 0$, $k=1,\dots,n$, we have:
\begin{equation}\label{eq:gamma-psi}
\frac{\d}{\d u_i} \Psi_d = E_{kk}\Psi_{d-1} + [\nd\gamma,E_{kk}] \Psi_d
\end{equation}
Here and below we assume that $\Psi_{-1}=0$ and we use $\gamma$ with arbitrary diagonal terms since they disappear in the commutator with $E_{kk}$. By $E_{kk}$ we denote the matrix unit, that is $(E_{kk})_{ij}:=\delta_{ik}\delta_{jk}$. This formula gives an expression for all elements of $\nd\gamma$ in terms of $\Psi_0$ and its derivatives.

We combine equations~\eqref{eq:r-psi} and~\eqref{eq:gamma-psi}:
\begin{align*}
\left.\frac{\d}{\d\epsilon}\right|_{\epsilon=0}\frac{\d}{\d u_i} \Psi_0(A\exp{\epsilon(r z^{-\ell})}) & = 
[\left.\frac{\d}{\d\epsilon}\right|_{\epsilon=0}\gamma (A\exp{\epsilon(r z^{-\ell})}),E_{kk}] \Psi_0(A) \\
&+ [\gamma (A),E_{kk}] 
\left.\frac{\d}{\d\epsilon}\right|_{\epsilon=0}\Psi_0 (A\exp{\epsilon(r z^{-\ell})})
\end{align*}
We denote $\left.\frac{\d}{\d\epsilon}\right|_{\epsilon=0}\gamma (A\exp{\epsilon(r z^{-\ell})})$ by $\delta\gamma$. Using that $\Psi^t(-z)\Psi(z)=\mathrm{Id}\cdot z^0$, we obtain the following equation:
\begin{align*}
& \frac{\d}{\d u_i}\left(\Psi_{\ell}r  
+ \sum_{q=0}^{\ell-1}(-1)^{\ell-q}\Psi_q r \Psi_{\ell-q}^t\Psi_{0}\right) = \\
&[\delta\gamma,E_{kk}] \Psi_0
+[\gamma,E_{kk}]\left(\Psi_{\ell}r+\sum_{q=0}^{\ell-1}(-1)^{\ell-q}\Psi_q r \Psi_{\ell-q}^t\Psi_{0}\right)
\end{align*}
Using equation~\eqref{eq:r-psi}, we see that 
\begin{align*}
[\delta\gamma, E_{kk}]\Psi_0 = [E_{kk}, \sum_{i=0}^{\ell-1} (-1)^i \Psi_{\ell-1-i} r \Psi_i] \Psi_0.   
\end{align*}
Since $\Psi_0$ is invertible, $\Psi_0\Psi_0^t=0$, we obtain equation~\eqref{eq:r-gamma}. 
Equation~\eqref{eq:s-gamma} can be proven in the same way, but in fact it is obvious from the definition of $\gamma$.
\end{proof}

\begin{example}\label{example} The simplest non-trivial deformation would be by an element $rz^{-1}$, where $r$ is an arbitrary symmetric matrix.
Denote by $\delta_r\gamma$ and $\delta_r\Psi_d$ the corresponding infinitesimal deformations. We have the following system of equations:
\begin{align} \label{eq:example}
\nd\delta_r\gamma & =-\nd \Psi_0 r \Psi_0^t \\ \notag
\delta_r\Psi_0 & = \Psi_1 r - \Psi_0 r \Psi_0^t \Psi_1 \\ \notag
\delta_r\Psi_1 & = \Psi_2 r - \Psi_0 r \Psi_0^t \Psi_2 \\ \notag
& \mbox{and so on.} 
\end{align}
\end{example}



\section{Special deformations}\label{sec3}

Shramchenko's deformations of Hurwitz Frobenius manifolds discussed in section~\ref{sec:shr} fits into a special case of example~\ref{example} that can be integrated explicitely.

\subsection{The input}
Consider a symmetric $n\times n$ matrix $\nd\gamma$ whose elements are functions in $u_1,\dots,u_n$. Let $\nd\gamma$ be a solution of the Darboux-Egoroff equations~\eqref{eq:DE}. There are two important geometric structures associated to the Frobenius structure corresponding to $\nd\gamma$. 

First, there is a solution of the commutativity equations~\cite{ShaZvo}, which is a symmetric $n\times n$ matrix $C=C(u_1,\dots,u_n)$ such that $dC\wedge dC=0$. In terms of multi-KP tau-functions, $C=\Psi_0^t\Psi_1$. 

Second, one can consider $\Psi_0$ itself. In geometric terms $\Psi_0$ is defined by the equation $dC=\Psi_0^t\cdot{diag}(du_1,\dots,du_n)\cdot\Psi_0$.
An alternative way to define $\Psi_0$ is the following. Consider the system of equations (it is equivalent to equation~\eqref{eq:gamma-psi}):
\begin{align}
\frac{\d \left(\Psi_0\right)_{ij}}{\d u_k}& =\gamma_{ik}\left(\Psi_0\right)_{kj}, & i\ne k,\\
\sum_{k=1}^n\frac{\d \left(\Psi_0\right)_{ij}}{\d u_k}& =0.&
\end{align}
Compatibility of this system of equations follows from the Darboux-Egoroff equations for $\nd\gamma$. This system of equations implies that $\d (\Psi_0^t\Psi_0)/\d u_k=0$, $k=1,\dots,n$, and $\Psi_0$ that we need is a particular solution of this system of equations such that $\Psi_0^t\Psi_0=\mathrm{id}$.

\subsection{Special deformations} \label{sec:special}
We consider a distribution in the tangent bundle of the moduli space of the solutions of the Darboux-Egoroff equations. It is given by the Givental-van de Leur tangent vectors of the type~\eqref{eq:example} described in example~\ref{example}. It is easy to see that, roughly speaking, a deformation of a particular solution of the Darboux-Egoroff equations is given by an ordinary differential equations of the infinite order. 

However, there is a special class of infinitesimal deformations that can be reduced to a finite order ODEs. We fix a positive integer $g\leq n/2$. Let $D$ be a $g\times n$ constant matrix of rank $g$ such that $DD^t=0$. Let us consider the distribution in the tangent bundle of the moduli space of the solutions of the Darboux-Egoroff equations given by the Givental-van de Leur tangent vectors of the type~\eqref{eq:example} described in example~\ref{example} with the matrix $r$ that can be represented as $r=D^t M D$, where $M$ is an arbitrary symmetric $g\times g$ matrix. In that case equation~\eqref{eq:example} can be reduced to an ODE of finite order.

\begin{proposition} Equation~\eqref{eq:example} for the matrix $r=D^t M D$ implies the following system of ODEs for $\nd\gamma$, $\omega:=D\Psi_0^t$, and $B:=DCD^t$:
\begin{align}\label{eq:ddt}
\nd \delta_M\gamma & = -\nd \omega^t M \omega; \\ \notag
\delta_M \omega & = -BM\omega; \\ \notag
\delta_M B & = -B M B.
\end{align}
\end{proposition}

\begin{proof} Direct computation. \end{proof}

In order to use this proposition for a particular $\nd\gamma$ without going back to the full multi-KP framework, we need an independent definitions of $\omega$ and $B$ in terms of $\gamma$ and $D$. We define $\omega$ as a $g\times n$-matrix-valued solutions of the equation
\begin{equation}\label{eq:defomega}
d\omega = \omega\cdot [diag(du_1,\dots,du_n),\nd\gamma]
\end{equation}
with the constant term $\omega|_{u=0}=D\Psi_0^t|_{u=0}$.
We define $B$ to be a $g\times g$-matrix-valued solution of the equation
\begin{equation}\label{eq:defB}
dB = \omega\cdot diag(du_1,\dots,du_n)\cdot \omega^t
\end{equation}
with the constant term $B|_{u=0}=DCD^t|_{u=0}$.

Equations~\eqref{eq:ddt} can be integrated explicitely in the case when $M$ is a constant matrix independent of $\nd\gamma$, $\omega$, and $B$. Indeed,
let us define $\nd\gamma(\epsilon)$, $\omega(\epsilon)$, and $B(\epsilon)$ by the following formulas: 
\begin{align}\label{eq:defform}
\nd\gamma(\epsilon) & := \nd\gamma-\nd\omega^t \epsilon M(1+\epsilon BM)^{-1} \omega; \\ \notag
\omega(\epsilon) & = (1+\epsilon BM)^{-1}\omega; \\ \notag
B(\epsilon) & = (1+\epsilon BM)^{-1}B.
\end{align}
(these formulas are defined in the domain where $(1+\epsilon BM)$ is invertible).

\begin{proposition} The matrices $\nd\gamma(\epsilon)$, $\omega(\epsilon)$, and $B(\epsilon)$ satisfy equations~\eqref{eq:defomega} and~\eqref{eq:defB} for any $\epsilon\geq 0$. They integrate the constant vector field determined by the matrix $M$, that is,
\begin{align*}
\nd\frac{\d \gamma(\epsilon)}{\d \epsilon} & =-\nd \omega(\epsilon)^t M \omega(\epsilon); & \nd\gamma(\epsilon)|_{\epsilon=0}& =\nd\gamma; \\
\frac{\d \omega(\epsilon)}{\d \epsilon} & =-B(\epsilon)M\omega(\epsilon); & \omega(\epsilon)|_{\epsilon=0}&=\omega; \\
\frac{\d B(\epsilon)}{\d \epsilon} & =-B(\epsilon) M B(\epsilon); & B(\epsilon)|_{\epsilon=0}&=B; \\
\end{align*}
\end{proposition}

\begin{proof} Direct computation. \end{proof}

\subsection{Shramchenko's formulas}

In this context, Shramchenko's formulas are a version of formulas~\eqref{eq:ddt} for $\epsilon=1$, with some appropriate changes. 
Let use $\nd\gamma$, $\omega$, and $B$ defined in section~\ref{sec:shr}. Equations~\eqref{eq:Rauch1}-\eqref{eq:Rauch3} imply Darboux-Egoroff equations for $\nd\gamma$ and equations~\eqref{eq:defomega}-\eqref{eq:defB}. Therefore, we are indeed have a system suitable for deformation given by~\eqref{eq:ddt} (the initial conditions for $\omega$ and $B$ depend on the choice of a particular point of a formal expansion). 
Indeed, let us substitute $\epsilon=1$ in equation~\eqref{eq:defform}. We have: 
\[
\nd\gamma(\epsilon)|_{\epsilon=1}=\nd\gamma(\epsilon)|_{\epsilon=0} - \nd\omega^t (M^{-1}+B)^{-1} \omega.
\] 
If we change the notations in order to replace $M^{-1}$ with $M$, we obtain exactly formula~\eqref{eq:shr}.

\begin{remark} One could obtain the same solution of the Darboux-Egoroff equations from a special deformation of the following triple: $\nd\tilde{\gamma}:=\gamma-\omega^tB\omega$, $\tilde{\omega}:=B^{-1}\omega$, and $\tilde{B}:=B^{-1}$. In that case some formulas would look a bit simpler.
\end{remark}

\begin{remark} Deformations of ``real doubles''~\cite{Shr1} of Hurwitz Frobenius manifolds fit into exactly the same scheme as we discuss in section~\ref{sec:special}. 
\end{remark}


\end{document}